\documentclass[letterpaper, 10 pt, conference]{ieeeconf}  

\IEEEoverridecommandlockouts                              

\usepackage{amsmath,amsfonts,amssymb,mathtools}
\usepackage{graphics,graphicx,epstopdf,float}
\usepackage[dvipsnames]{xcolor}
\usepackage[normalem]{ulem}
\graphicspath{{FIGURES/}}
\usepackage{caption,subcaption}
\usepackage{comment}

\usepackage[colorinlistoftodos]{todonotes}  

\newcommand{\nbresponse}[1]{\\{\textbf{\color{blue}#1}}}

\newcommand{\pd}[2]{\frac{\partial #1}{\partial #2}}
\newcommand{\mtx}[1]{\begin{bmatrix}#1\end{bmatrix}}

\newtheorem{definition}{Definition}
\newtheorem{theorem}{Theorem}

\title{\LARGE \bf
Sensor Placement on a Cantilever Beam Using Observability Gramians}

\author{Natalie L. Brace$^{1}$, Nicholas B. Andrews$^{1}$, Jeremy Upsal$^{2}$, and Kristi A. Morgansen$^{1}$
\thanks{$^{1}$The authors are with the Department of Aeronautics and Astronautics, University of Washington, Seattle, WA 98195-2400. $^{2}$The author is with the Department of Applied Mathematics, University of Washington, Seattle, WA 98195-3925. Contact {\tt\small nbrace@uw.edu,
nian6018@uw.edu, jupsal@uw.edu, morgansn@uw.edu}. This material is based upon work funded by the Joint Center for Aerospace Technology Innovation.}}%

\begin{document}

\onecolumn

\twocolumn
\maketitle
\thispagestyle{empty}
\pagestyle{empty}

\setcounter{page}{1}

\begin{abstract}
Working from an observability characterization based on output energy sensitivity to changes in initial conditions, we derive both analytical and empirical observability Gramian tools for a class of continuum material systems.
Using these results, optimal sensor placement is calculated for an Euler-Bernoulli cantilever beam for the following cases:  analytical observability for the continuum system and analytical observability for a finite number of modes.  Error covariance of an Unscented Kalman Filter is determined for both cases and compared to randomly placed sensors to demonstrate effectiveness of the techniques.
\end{abstract}

\section{INTRODUCTION} \label{sec:intro}

Structural health monitoring is the process of using sensors, either in situ or remote, in combination with analytical and empirical tools and algorithms to assess the dynamic material and geometric properties of structures and systems.  The manufacturing processes for most such continuum systems present challenges with incorporating sensors or with replacing them in the event of failure.  For example, aircraft wings manufactured from composites can have sensors embedded during manufacturing, but replacing them after the fact would damage the material.  Advection-diffusion processes of ocean properties such as temperature, salinity and particulate operate over such large physical scales that continual sensing across the entire process at high resolution is prohibitive.  A typical question that is addressed in monitoring such systems is where and when to place sensors and what tools to use to process the information from the sensors.  Generally, either a finite dimensional modal approach is taken or, in some simpler systems, analysis of the infinite dimensional continuum system can be addressed.  Here, we are specifically interested in extending recent developments in empirical methods for finite dimensional nonlinear systems to continuum systems to eventually consider models where complexities in the dynamics such as nonlinearities lead to suboptimal results from existing methods.  We are particularly motivated by applications in strain sensing in insect flight, monitoring of aircraft wing flexure as related to ride quality for active control, and deployment of sensor packages in oceanographic monitoring networks.  In order to ground and assess our studies, in this work we focus on the analytically tractable Euler-Bernoulli cantilever beam.

The current approaches to obtaining dynamic process information for such systems is generally based on finite dimensional approximations of the full continuum system using modal analysis or on finite element models. Controllability and observability analysis on general modal systems is considered by Yang in \cite{Yang1995}.  Their work determined that the systems considered are observable with a single sensor so long as it is located in an appropriate place. Analytical controllability and observability of infinite dimensional continuum systems has been considered in \cite{ElJai1988} and \cite{Curtain2012}, however, the results are general and not closed form for the class of systems considered here.
In situations where analytical models are intractable, the simulation-based empirical observability Gramian \cite{Krener2009} was developed and has been used to assess (un)observability indices for finite dimensional systems.  In this technique, the initial conditions of the simulations are perturbed by particular amounts, and in the limit as the perturbations become small, the empirical observability Gramian becomes exactly the analytical observability Gramian.

Optimal sensor placement using an exact analytical observability Gramian was assessed by Georges \cite{Georges2017} for a class of advection-diffusion partial differential equation (PDE) systems. The class of systems considered here, however, does not immediately admit the tools used in \cite{Georges2017}, although with the addition of a damping term the Euler-Bernoulli beam would admit a Riesz basis \cite{Guo2001} similar to Georges' system. In another approach, a non-classical micro-beam is studied to find that measuring the moment at the root of the beam produces an exactly observable system \cite{Edalatzadeh2017}.  For truncated finite dimensional representations of such mechanical continuum systems,  \"{O}rtel \cite{Oertel2013} considered optimal placement of an accelerometer, gyro, and strain gauge on a cantilever beam with respect to a penalty function designed to minimize the contributions of the fourth and fifth modes while promoting the first three modes. Work by  Menuzzi \cite{Menuzzi2017} used finite element analysis modeling to place piezoelectric elements on a cantilever beam, optimizing the trace of the linear observability Gramian via topology optimization that measures sensitivity through the derivative of the Lyapunov equation. In related work with insects \cite{Hinson2015GyroscopicSensitivity},  Hinson placed bending and shear strain sensors on the wings of a hawkmoth to maximize observability of inertial rotations with a cost function based on condition number and the inverse of the minimum eigenvalue of the empirical observability Gramian.    

The primary contribution of the work here is the extension of empirical observability Gramian tools from finite dimensional systems to continuum systems.  To facilitate this development, we consider a standard Euler-Bernoulli beam equation for which analytical results can be exactly calculated and compared to the empirical results.  The resulting analytical observability parallels the developments of \cite{Georges2017} with the benefit of the parallel empirical tool being viable for systems to which those analytical tools do not apply.  We demonstrate the usefulness of the results via the task of optimal sensor placement with demonstration of improvement in error covariance of an Unscented Kalman Filter (UKF).
%

The paper is organized as follows: the model of the Euler-Bernoulli beam is presented in Section \ref{sec:model}, followed by background on observability in Section \ref{sec:observabilityTools}. The analytical and empirical Gramians for the PDE system are developed in Section \ref{sec:analytical}. In Section \ref{sec:numerical}, numerical results are presented for sensor placement and estimation of a vibrating beam, and conclusions are presented in Section \ref{sec:future}.


\section{System Models}\label{sec:model}

\begin{figure}
    \centering
    \includegraphics[height=1in]{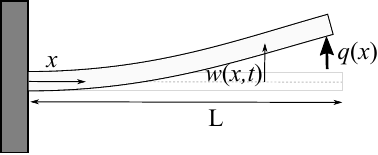}
    \caption{Model of a cantilever beam.}
    \label{fig:model}
\end{figure}

The transverse deflection, $w(x,t)$, of an  Euler-Bernoulli cantilever beam clamped at $x=0$ and free at $x=L$  under free vibration is described by the PDE
\begin{equation}
\begin{gathered}\label{eq:freePDE}
    \pd{^2}{x^2} \left( E(x)I(x) \pd{^2 w(x,t)}{x^2}\right) = -\mu \pd{^2 w(x,t)}{t^2} \\
    w(0,t) = w_x(0,t) = w_{xx}(L,t) = w_{xxx}(L,t) = 0, 
\end{gathered}
\end{equation}
with initial condition $w_0(x) = w(x,0)$, where $w_{ x\ldots x}(a,t) = \left. \pd{}{x} \left( \cdots \left(\pd{w(x,t)}{x}\right) \right)\right|_{x=a}$,  $E(x)$ is the elastic modulus, $I(x)$ is the second moment of inertia, and $\mu$ is the mass per unit length \cite{Gatti2014}.

This system can be solved using separation of variables by letting $w(x,t) = \phi(x)\eta(t)$, with the resulting continuum system solution, $\Sigma_\infty$, given by
\begin{equation}
\Sigma_\infty: \qquad
\begin{aligned}
w(x,t) &= \sum_{i=1}^\infty \phi_i(x) \eta_i(t), \\
y(x_\ell,t) &= h_\ell \sum_{i=1}^\infty \left. \pd{^2 \phi(x)}{x^2}\right|_{x=x_\ell} \eta_i(t), 
\end{aligned} \label{eqn:pdesoln}
\end{equation}
where the measurements are the strain taken at location $x_\ell$ at which point the beam's height from the neutral axis is equal to $h_\ell$. The spatial mode shapes of the beam, $\phi_i(x)$, and temporal modal coefficients, $\eta_i(t)$, satisfy the ordinary differential equations (ODE)
\begin{align}
 \pd{^2}{x^2} \left( E(x)I(x) \pd{^2 \phi_i}{x^2}\right) = \omega_i^2 \phi_i, \quad
 \pd{^2 \eta_i}{t^2} = -\mu \eta_i \omega_i^2.
\end{align}
\color{black}
 
Under the assumption of a uniform cross section, i.e. $E$ and $I$ are constant, the mode shapes are given by 
\begin{align}
    \phi_i(x) &= \cosh{b_i x} - \cos{b_i x} + f_i(L)(\sin{b_i x} - \sinh{b_i x}), \nonumber \\
    & \qquad \qquad f_i(L) = \frac{\cos{b_i L} + \cosh{b_i L}}{\sin{b_i L} + \sinh{b_i L}}, \label{eq:modeShapes}
\end{align}
where $b_i$ satisfies $\cos(b_i L) \cosh(b_i L) = -1$ \cite{Gatti2014}; the first ten mode shapes are shown in Fig. \ref{fig:mode_shapes}. The time-dependence of the modal coefficients is
\begin{align}
    \eta_i(t) = \alpha_{1,i} \cos(\omega_i t) + \alpha_{2,i} \frac{\sin(\omega_i t)}{\omega_i}, ~ \omega_i = b_i^2 \sqrt{\frac{EI}\mu} \label{eq:eta_i}
\end{align}
where the generalized Fourier coefficients $\alpha_{1,i}$ and $\alpha_{2,i}$ are determined by the initial conditions with respect to shape and velocity as given in section \ref{sec-continuumObsvGram}.
%

To approximate $\Sigma_\infty$ as a finite-dimensional system, the summations are truncated and the states of the system are defined as the first $n_\phi$ modes and their derivatives, 
  $H = [\eta_1  ~~ \eta_2 ~~  \cdots ~~  \eta_n  ~~ \dot{\eta}_1  ~~ \dot{\eta}_2  ~~ \cdots ~~  \dot{\eta}_n]^\top \in \mathbb{R}^{n},$
where ${n = n_\eta  n_\phi}$ with $n_\eta=2$ to account for the second order ODE. The dynamics  for the coefficients $\eta_i$ can then be written as
\begin{equation}
\Sigma_n: \qquad
\begin{aligned}
\dot{H} &= AH, 
\quad A = \mtx{0 & I_{n_\phi} \\ \Omega_{n_\phi} & 0},  \\
\mathbf{y} &= CH, 
\end{aligned} \label{eq:linSys}
\end{equation}
where $I_{n_\phi}$ is the $n_\phi \times n_\phi$ identity matrix, $\Omega_{n_\phi} =\text{diag}(\bar{\boldsymbol{\omega}})$ with  $\bar{\boldsymbol{\omega}} = \mtx{-\omega_1^2 & -\omega_2^2 & \cdots &-\omega_{n_\phi}^2}$, and each row $c_\ell$ of measurement matrix $C \in \mathbb{R}^{p \times n}$ is defined as the strain at sensor location $x_\ell$ given by
\begin{align}
 &c_\ell = h_\ell \mtx{\pd{^2 \phi_1(x_{\ell})}{x^2} & \pd{^2 \phi_2(x_{\ell})}{x^2} ~ \cdots ~ \pd{^2 \phi_{n_\phi}(x_{\ell})}{x^2} & \mathbf{0}_{1 \times n_\phi}}. \label{eq:c_i} 
\end{align}
\\

\begin{figure}
    \centering
    \includegraphics[width =0.5\textwidth]{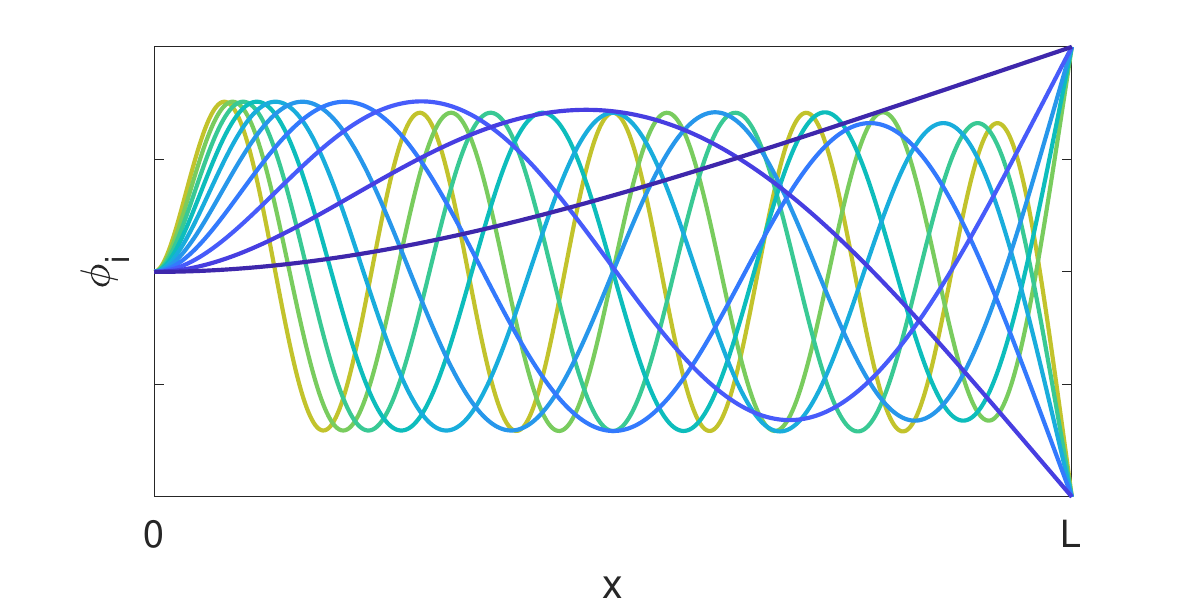}
    \caption{First ten mode shapes for an Euler-Bernoulli beam fixed at the left end.}
    \label{fig:mode_shapes}
\end{figure}

\section{Observability Tools} \label{sec:observabilityTools}

Observability describes the feasibility of uniquely determining an initial state of a system based on measurements of the system over a finite time interval. If the unknown initial state $\mathbf{x}_0$ can be uniquely determined in an open neighborhood of $\mathbf{x}_0$ from the outputs $y$, then the system is {\em weakly observable. }Linear systems can be evaluated analytically using the standard observability matrix and Gramian, whereas a computational approach can be taken to generate an empirical observability Gramian for nonlinear and analytically intractable systems. In combination with observability tools, measures of observability based on those tools can tell us whether or not a system is observable as well as provide the basis for a metric for optimal placement of sensors.

\subsection{Linear Observability for Finite-Dimensional Systems}

In linear systems, the observability matrix, $\mathcal{O}$, is obtained by differentiating the output $\mathbf{y}$ in (\ref{eq:linSys}) with respect to time and collecting terms multiplying the state:
\begin{align}
    \mathcal{O} = \mtx{
    C \\ CA \\ \vdots \\ CA^{n-1} }.
    \label{eq:Omtx}
\end{align}
A linear system is then observable if and only if $\mathcal{O}$ is full rank \cite{chen1999linear}.

Another tool to determine observability is the observability Gramian, which quantitatively captures the sensitivity of the measurements to a change in the initial conditions. The standard expression of the observability Gramian for a linear time-invariant (LTI) system is 
\begin{align}
    W_o(t) &= \int_{0}^{t} e^{A^\top \tau}C^\top C e^{A \tau} d\tau \label{eq:Wlinear}
    \end{align}
which can also be written as    
\begin{align}
   W_{o}(t) &= \sum_{\ell=1}^{p}\int_{0}^{t}  \pd{y_\ell(\tau)}{\mathbf{x}_0} \left( \pd{y_\ell(\tau)}{\mathbf{x}_0} \right)^\top d\tau,
   \label{eqn:vargramian}
\end{align}
where we note that the derivative of $y_\ell$ with respect to $\mathbf{x}_0$ is a column vector.
If $W_o(t) \in \mathbb{R}^{n \times n}$ is nonsingular for some $t>0$, the system will be observable \cite{chen1999linear}. Furthermore, the eigenvector associated with the largest (smallest) eigenvalue indicates the mode that is most (least) observable.

\subsection{Linear Observability for Infinite-Dimensional Systems}

The observability Gramian for a continuum system developed in this paper will be based off of extending \eqref{eqn:vargramian}, however, 
 definitions of observability for infinite-dimensional systems similar to those of LTI systems and are included here for completeness. As presented in \cite{Curtain2012}, a system
\begin{align}
    \dot{z}(t) = A_\infty z(t), ~ y(t) = C_\infty z(t), ~ z \in Z_\infty, y\in Y_\infty \label{eq:infDimLinSys}
\end{align}
where $Z_\infty$ and $Y_\infty$ are Hilbert spaces, $A_\infty$ is the infinitesimal generator of the strongly continuous semigroup $T(t)$ on $Z_\infty$, and $C_\infty$ is a bounded linear operator from $Z_\infty$ to $Y_\infty$.

Then the observability map of the system in \eqref{eq:infDimLinSys} on $[0,\tau]$ is given as $\mathcal{C}^{\tau}z := C T(\cdot) z$. The system in \eqref{eq:infDimLinSys} is \emph{approximately observable} on $[0,\tau]$ if the initial state is uniquely determined by the output $L_2([0,\tau]; Y)$, in which case ker $\mathcal{C}^\tau = \{0\}$ and the observability Gramian, $W_C^\tau=\mathcal{C}^{\tau*}\mathcal{C}^\tau$, is greater than zero \cite{Curtain2012,Georges2017}.

\subsection{Empirical Observability Gramian}

The empirical observability Gramian provides a method to approximate $W_o$ for nonlinear systems or for linear systems where the Gramian is difficult to calculate. The sensitivity of the measurements to changes in the initial conditions is captured by perturbing the initial states by a small value, $\epsilon$, and compiling the results as
\begin{equation*}
    W_o^\epsilon(t)=\frac{1}{4\epsilon^2} \int_{0}^{t} 
    \Delta Y^\top \Delta Y d\tau
    \label{eqn:empGram}
\end{equation*}
where $\Delta Y = \mtx{\Delta y^{\pm1}(\tau) & \Delta y^{\pm2}(\tau) & \cdots & \Delta y^{\pm n}(\tau) }$ is a $\mathbb{R}^{1 \times n}$ vector comprised of the differences of the scalar outputs, $\Delta y^{\pm i}(t) = y_i^+ - y_i^-$,  that result from perturbing the initial condition $\mathbf{x}_0$ by $\pm \epsilon \hat{\mathbf{e}}_i$.
If the empirical observability Gramian, $W_o^\epsilon(t) \in \mathbb{R}^{n \times n}$, is full rank at the limit $\epsilon\to0$, then the system is weakly observable at $\mathbf{x}_0$ \cite{Powel2015EmpiricalControl}. 

\subsection{Measures of Observability}
The following measures provide a method of quantifying the degree of observability based on the observability Gramian and allow for optimization over a potential sensor set \cite{Krener2009}.

\subsubsection{Local unobservability index}
This metric is the reciprocal of the minimum eigenvalue and provides a measure of the least observable mode
\begin{align*}
    J_{\nu}(W_o) = \frac{1}{\lambda_{min}(W_o)} = \nu.
\end{align*}
The smaller this number is, the better conditioned the inversion of the map from states to measurements will be.

\subsubsection{Local estimation condition number}
This measure is the ratio between the largest and smallest eigenvalues
\begin{align*}
    J_{\kappa}(W_o) = \frac{\lambda_{max}(W_o)}{\lambda_{min}(W_o)} = \kappa.
\end{align*}
The closer to one this measure is, the more balanced the information is available from the sensor output. The condition number is not a metric and must be used with some caution as it may prioritize minimizing $\lambda_{max}$.

For this paper, the objective function is a combination of the above measures with a weighting of $w = 5$, chosen empirically to provide a balance between the metrics,
\begin{align}
    J(W_o) = \kappa + w \nu. \label{eq:objFunc}
\end{align}

\subsection{Optimal Sensor Placement}
With binary sensor activation variables $\alpha_i \in \{0,1\}$ that indicate if a sensor is in use, define the total observability Gramian as the sum $\tilde{W}(\boldsymbol{\alpha}) = \sum_{i=1}^{n_p} W_i \alpha_i$ where $n_p$ is the number of potential sensor locations. 
The optimal placement for $p$ sensors given objective function $J(\tilde{W}_o)$ is 
\begin{eqnarray}
 \min_{\alpha} & & J(\tilde{W}_o) \nonumber  \\
\text{subject to}  & & \sum_{i=1}^{n_p} \alpha_i \leq p \nonumber \\
                   & & \alpha_{i} \in \{0,1\} \quad \forall i. \nonumber  \nonumber
\end{eqnarray} 
This formulation is, however, a non-convex mixed integer program, making it difficult to solve. By easing the requirement that the sensor activation variable, $\alpha_i \in \{0,1\}$, be binary and instead requiring a value between zero and one, $0 \leq a_i \leq 1$, the constraints become convex. If the objective function is convex with respect to the activation variables, the optimization problem is then convex. The measures discussed in the previous section are convex ($\nu$) or quasiconvex ($\kappa$) with respect to the variables $a_i$ \cite{Boyd2004}. This relaxation from binary to continuous sensor activation variables makes the solution tractable, however it is generally sub-optimal.

For a balance of condition number and inverse of the minimum eigenvalue, the following optimization problem can be posed:
\begin{eqnarray}
 \min_{\mathbf{a}, \kappa, \nu} & & \kappa + w \nu  \nonumber\\
\text{subject to}  & & \tilde{W}(\bar{\mathbf{a}}) - I \succeq 0  \nonumber \\
                   & & \kappa I - \tilde{W}(\bar{\mathbf{a}}) \succeq 0  \label{eq:optProb} \\
                   & & 0 \leq \bar{a}_{i} \leq \nu \quad \forall i  \nonumber \\
                   & & \sum_{i=1}^{n_p} \bar{a}_i \leq p \nu. \nonumber
\end{eqnarray}
The conditions are equivalent to $\lambda_{min} I \preceq \tilde{W}(\mathbf{a}) \preceq \lambda_{max} I$ along with the conditions on $a_i$ after normalizing the inequalities by $\lambda_{min}$ and applying the change of variables $\bar{a}_i = \nu  a_i$.

\section{Analytical Results}\label{sec:analytical}

Using the tools from the previous section, we provide here analytical results for the Euler-Bernoulli cantilever beam beginning with developing the analytical and empirical observability Gramians of the PDE system, $\Sigma_\infty$. 
Then we show that the ODE system $\Sigma_n$ is (technically) observable with a single sensor.

\subsection{Continuum Analytical Observability Gramian} \label{sec-continuumObsvGram}
For the linear ODE system, the observability Gramian can be defined through the derivative of the measurements with respect to the initial conditions as in (\ref{eq:Wlinear}) and (\ref{eqn:vargramian}). To accommodate the PDE system structure, we replace the derivative in (\ref{eqn:vargramian}) with the first variation. 
For the class of systems considered here, the full solution can be expressed as
\begin{align*}
    w(x,t) &= \sum_{i=1}^\infty \phi_i(x) \left( \alpha_{1,i} \cos(\omega_i t) + \alpha_{2,i} \frac{\sin(\omega_i t)}{\omega_i} \right)
\end{align*}
where $\alpha_{1,i}= \frac{1}{c_i}\int_0^L w_0(x)  \phi_i(x) dx$ with initial displacement $w_0(x) = w(x,0)$, $\alpha_{2,i}=\frac{1}{c_i}\int_0^L \dot{w}_0(x)  \phi_i(x) dx$ with initial velocity $\dot{w}_0(x) = \dot{w}(x,0)$, and $c_i = \int_0^L \phi_j(x)^2 dx$ is a normalization factor.
Perturbing the initial condition of the displacement, $w_0(x)$, by some function, $\epsilon f(x)$, for a small value of $\epsilon$, the modal coefficients, $\eta_i(t)$ become
\begin{align*}
    \eta_i(t)^{+ f} &= \eta_i(t) +  \epsilon C_i^f \cos(\omega_i t), 
\end{align*}
where $C_i^f = \frac{1}{c_i}\int_0^L f(x)  \phi_i(x) dx$. The perturbed general solution is then
\begin{align*}
    w(x,t)^{+ f} &= \sum_{i=1}^{\infty} \phi_i(x) \left( \alpha_{1,i}^{+f} \cos(\omega_i t) + \alpha_{2,i} \frac{\sin(\omega_i t)}{\omega_i} \right) \\
    &= w(x,t) + \epsilon \sum_{i=1}^{\infty} \phi_i(x) C_i^f \cos(\omega_i t).
\end{align*}
Since the system is second order with respect to time, the process is repeated by perturbing the initial velocity $\dot{w}_0(x)$ by $\epsilon g(x)$. The perturbed modal coefficients in this case become
\begin{align*}
  \eta_i(t)^{+ g} &= \eta_i(t) + \epsilon C_i^g \frac{\sin(\omega_i t)}{\omega_i}
\end{align*}
with $C_i^g = \frac{1}{c_i}\int_0^L  g(x)  \phi_i(x) dx$.
The resulting perturbations in the measurements are calculated using the modal coefficients as
\begin{align}
   y_\ell(t)^{+ f} &= h_\ell \sum_{i=1}^{\infty} \phi_{i,xx}(x_\ell) \left(\eta_i(t) + \epsilon C_i^f \cos(\omega_i t) \right) \label{eq:y_pm_f} \\
    y_\ell(t)^{+ g} &= h_\ell \sum_{i=1}^{\infty} \phi_{i,xx}(x_\ell) \left(\eta_i(t) + \epsilon C_i^g \frac{\sin(\omega_i t)}{\omega_i} \right). \label{eq:y_pm_g} 
\end{align}
\begin{definition}
 A \emph{functional}, $K$, is a map between functions, e.g., $K:z(t)\to K[z](t)$. 
 \end{definition}
 \begin{definition}
The \emph{first variation} of a functional $K[z](t)$ is a functional that maps the perturbation, $h$, to $\delta K[z;h] = \lim_{\epsilon \to 0} \frac{K[z + \epsilon h] - K[z]}{\epsilon}$  \cite{gelfand1963calculus}.
\end{definition}

Using this notation, we may view $w(x,t)$ as a functional mapping initial conditions, $[w_0,\dot{w}_0]$, to a solution of the initial boundary value problem \eqref{eq:freePDE}. In other words, $w[w_0,\dot{w}_0](x,t)$ is a functional. Similarly, we may view the measurements $y_\ell[w_0, \dot{w}_0](t)$ as functionals. Then the first variations of the measurement $y_\ell[w_0,\dot{w}_0](t)$ with respect to the functions $f(x)$ and $g(x)$ are

\begin{align}
    \delta y_\ell[w_0,\dot{w}_0; f](t) &= \lim_{\epsilon \to 0} \frac{y_\ell^{+f} - y_\ell}{\epsilon} \nonumber \\
    &=  h_\ell \sum_{i=1}^{\infty} \phi_{i,xx}(x_\ell) C_i^f \cos(\omega_i t) \\
    \delta y_\ell[w_0, \dot{w}_0; g](t) &= \lim_{\epsilon \to 0} \frac{y_\ell^{+g} - y_\ell}{\epsilon} \nonumber \\
    &=  h_\ell \sum_{i=1}^{\infty}\phi_{i,xx}(x_\ell) C_i^g \frac{\sin(\omega_i t)}{\omega_i}.
\end{align}
For the Gramian, the perturbation function is defined as each of the mode shapes in the following theorem.

\begin{theorem}[Continuum analytical observability Gramian] The observability Gramian, $W_\infty$, for system $\Sigma_\infty$ is given by 
\begin{equation}
    W_\infty(t,x_\ell) =  \int_0^t \delta W_\infty^\top \delta W_\infty d\tau \in \mathbb{R}^{n_\eta \times n_\eta},
\end{equation}
 where 
\begin{align*}
    \delta W_\infty^\top &= \mtx{ h_\ell \sum_{j=1}^{\infty} \phi_{j,xx}(x_\ell) \cos(\omega_j \tau) \\  h_\ell \sum_{j=1}^{\infty} \phi_{j,xx}(x_\ell) \frac{\sin(\omega_j \tau)}{\omega_j}~~ }.
\end{align*}

\end{theorem}
\begin{proof}
We restrict the class of perturbations, $f(x)$ and $g(x)$, to a space whose basis is spanned by the mode shapes, $\phi_j$. Since our system is linear, we may then determine the sensitivity of the output to changes in the initial conditions by perturbing by each of the mode shapes. 
Beginning with displacement, define $f(x) = \phi_j(x)$. Then the coefficient $C_i^f$ becomes
\begin{align*}
    C_i^f = \frac{1}{c_i}\int_0^L \phi_j(x) \phi_i(x) dx =
    \begin{cases}
    1 & i = j \\
    0 & i \neq j,
    \end{cases}
\end{align*}
with the same results for $C_i^g$ when $g(x)$ is defined as $\phi_j(x)$. Then the variations $\delta y_{\ell}$ with respect to $f(x)$ and $g(x)$ simplify to
\begin{align*}
\delta y_\ell^f[w_0, \dot{w}_0; \phi_j](t) &=  h_\ell \phi_{j,xx}(x_\ell) \cos(\omega_j t) \\
\delta y_\ell^g[w_0, \dot{w}_0; \phi_j](t) &=  h_\ell\phi_{j,xx}(x_\ell) \frac{\sin(\omega_j t)}{\omega_j}.
\end{align*}
To include the sensitivity of the output to each of the modes, these are then summed together
\begin{align*}
   \delta Y_1 =   \sum_{j=1}^{\infty} \delta y_\ell^f[w_0, \dot{w}_0; \phi_j](t) \\
   \delta Y_2 =  \sum_{j=1}^{\infty} \delta y_\ell^g[w_0, \dot{w}_0; \phi_j](t).
\end{align*}
Defining $\delta W_\infty = \mtx{\delta Y_1 & \delta Y_2}$, we have the desired result.
\end{proof}
We note that our continuum observability Gramian agrees with other work on PDE observability \cite{Georges2017}. Replacing their nonlocal observation operator \cite[(40)]{Georges2017} with our local observation operator and accounting for appropriate changes for the eigenfunctions of the spatial operator, we see that their $\mathcal{C}^T u$ is equivalent to one component of our $\delta W_{\infty}^T$. This difference is expected since their PDE is scalar whereas ours is second-order in time. 

\subsection{Continuum Empirical Observability Gramian}
A key feature of the approach we took to finding the analytical observability Gramian is that it facilitates the construction of an empirical approach for situations where an analytical result is intractable.  
The empirical method is constructed by extending the  empirical Gramian technique for an ODE system in which each initial state in $\mathbf{x}_0$ is perturbed and simulated twice (for $\pm \epsilon \hat{\mathbf{e}}_i$).  In the continuum framework the initial conditions for the displacement, $w_0(x)$, and velocity, $\dot{w}_0(x)$, are perturbed instead by the mode shapes, $\pm \epsilon \phi_i(x)$, and the resulting changes to the measurements are calculated. 
This process is essentially the same as for the analytical Gramian, but it is important to show the process for this analytically tractible system so it is understood for systems that will rely on simulation; indeed, since the specific PDE we consider here is linear, the results are identical.

\begin{theorem}[Continuum empirical observability Gramian] The empirical observability Gramian, $W_\infty^\epsilon$, for system $\Sigma_\infty$ is given by 
\begin{equation}
    W_\infty^{\epsilon}(t,x_\ell) = \frac{}{}\int_0^t \Delta Y_\infty^\top \Delta Y_\infty d\tau  \in \mathbb{R}^{n_\eta \times n_\eta},
\end{equation}
 where 
\begin{align*}
    \Delta Y_\infty^\top &= \mtx{ \sum_{J=1}^{\infty}h_\ell \phi_{J,xx}(x_\ell) \cos(\omega_i \tau) \\ \sum_{J=1}^{\infty} h_\ell \phi_{J,xx}(x_\ell) \frac{\sin(\omega_i \tau)}{\omega_J} }.\\
\end{align*}
\end{theorem}
\begin{proof}
For the definition of the empirical Gramian, assume the beam can be perturbed a small amount exactly by an individual mode shape, $j$, that is 
\begin{align*}
    w_0^{\pm j}(x) &= w_0(x) \pm \epsilon_{1,j} \phi_j(x) \\ 
    \dot{w}_0^{\pm j}(x) &= \dot{w}_0(x) \pm \epsilon_{2,j} \phi_j(x).
\end{align*}
Due to the orthogonality of the mode shapes, the $\alpha_{k,i}$ components of the modal coefficients will in turn be perturbed by $\epsilon_{k,j}$
\begin{align*}
  \alpha_{1,i}^{\pm j} &= \frac{1}{c_i}\int_0^L \big(w_0(x) + \epsilon_{1,j} \phi_j(x) \big) \phi_i(x) dx\\
  &= \begin{cases}
     \alpha_{1,i} & i \neq j\\
     \alpha_{1,j} \pm \epsilon_{1,j} & i = j
  \end{cases}\\
  \alpha_{2,i}^{\pm j} &= \frac{1}{c_i}\int_0^L \big(\dot{w}_0(x) + \epsilon_{2,j} \phi_j(x) \big) \phi_i(x) dx \\
    &= \begin{cases}
     \alpha_{2,i} & i \neq j\\
     \alpha_{2,j} \pm \epsilon_{2,j} & i = j,
  \end{cases}
\end{align*}
so that the perturbed modal coefficients become
\begin{align*}
    \eta_i(t)^{\pm \epsilon_{1,j}} &= 
    \begin{cases}
    \eta_i(t) & i \neq j\\
    \eta_j(t) + \epsilon_{1,j} \cos(\omega_j t) & i = j
    \end{cases}\\
    \eta_i(t)^{\pm \epsilon_{2,j}} &=
    \begin{cases}
    \eta_i(t) & i \neq j\\
    \eta_j(t) + \epsilon_{2,j} \frac{\sin(\omega_j t)}{\omega_j} & i = j.
    \end{cases}    
\end{align*}
The measurements then become
\begin{align*}
    y(x_\ell,t)^{\pm \epsilon_{kj}} &= \left(\sum_{i=1}^{\infty} \phi_{i,xx}(x_\ell) \eta_i(t)\right) \pm \epsilon_{kj} \phi_{j,xx}(x_\ell) \beta_{kj}(t), \nonumber \\
    &\beta_{kj(t)} = \begin{cases}
    \cos(\omega_j t) & k = 1 \\
    \frac{\sin(\omega_j t)}{\omega_j} & k = 2
    \end{cases}
\end{align*}
so we have
\begin{align*}
    \Delta y(x_\ell,t)^{\pm \epsilon_{k,j}} 
    =  2 \epsilon_{k,j} h_\ell \phi_{j,xx}(x_\ell) \beta_{kj}(t).
\end{align*}
To generate the Gramian, each of the $\Delta y(x_\ell,t)^{\pm \epsilon_{k,j}}$ terms are divided by $2 \epsilon_{k,j}$ (since the empirical Gramian is calculated as a central difference) and summed together to form $\Delta Y_\infty$. 
\end{proof}

Since $\Sigma_\infty$ is a linear system, the $\epsilon_{kj}$ terms cancel so the analytical and empirical versions match exactly; in a nonlinear system, this would not be the case, and the matching would rely on the limit of $\epsilon_{kj}$ becoming small.

\subsection{Single Sensor Observability}
The truncated system, $\Sigma_n$, is LTI, so the analytical observability Gramian can be calculated with \eqref{eq:Wlinear} and the observability matrix with \eqref{eq:Omtx}. 

\begin{theorem}[Single sensor analytical observability] \label{thm:singleSensor}
The system $\Sigma_n$ is observable with a single sensor measurement if and only if that sensor is not located at a zero of the second derivative of any mode shape with respect to $x$.
\end{theorem}

\begin{proof}
Assume we have a single measurement and denote $c'$ as first $n$ elements of $c_i$ in \eqref{eq:c_i},
so $y = \mtx{c' & \mathbf{0}}$ and the observability matrix \eqref{eq:Omtx} can be calculated as
\begin{align}
    \mathcal{O} 
    &=\mtx{c' & \mathbf{0} \\
           \mathbf{0} & c' \\
           \langle c',  \bar{\boldsymbol{\omega}}\rangle & \mathbf{0} \\
           \mathbf{0} & \langle c',  \bar{\boldsymbol{\omega}} \rangle \\
           \langle c', \bar{\boldsymbol{\omega}}.^2 \rangle & \mathbf{0} \\
           \mathbf{0} & \langle c', \bar{\boldsymbol{\omega}}.^2 \rangle \\
           \vdots & \vdots \\
           \langle c',  \bar{\boldsymbol{\omega}}.^{n-1} \rangle & \mathbf{0} \\
           \mathbf{0} & \langle c', \bar{\boldsymbol{\omega}}.^{n-1} \rangle
           }, \\
 \boldsymbol{\omega}.^a &=  \mtx{(-\omega_1^2)^a & (-\omega_2^2)^a & \cdots & (-\omega_n^2)^a} \nonumber
\end{align}
The rows of $\mathcal{O}$ can be reorganized to form a block diagonal matrix
\begin{align*}
\mathcal{O} = \mtx{\mathcal{O}_C & 0 \\ 0 & \mathcal{O}_C},
\end{align*}
so if $\mathcal{O}_C \in \mathbb{R}^{n_\phi \times n_\phi}$ is full rank, the system will be observable. To simplify notation let $d_j = -\omega_j^2$ and $p_k = \pd{^2\phi_k(x_{\ell})}{x^2}$; then 
   \begin{align*}
  \mathcal{O}_C 
    &= h_\ell \mtx{p_1 & p_2 & \cdots & p_{n_\phi} \\
         p_1 d_1 & p_2 d_2 & \cdots & p_{n_\phi} d_{n_\phi} \\
         p_1 d_1^2 & p_2 d_2^2 & \cdots &  p_{n_\phi} d_{n_\phi}^2 \\
         \vdots \\
          p_1 d_1^{{n_\phi}-1} &  p_2 d_2^{{n_\phi}-1} & \cdots & p_{n_\phi} d_n^{{n_\phi}-1}}.
\end{align*}
Note that $\mathcal{O}_C = h_\ell V^T P$, where $P = \text{diag}(p_1, p_2, \ldots, p_{n_\phi})$ and $V$ is the Vandermonde matrix in the $d_j$ with $\det(V) = \prod_{1\leq i \leq j \leq {n_\phi}} (d_j-d_i)$, therefore
\begin{align*}
    \det(\mathcal{O}_C) &= h_1 \det(V) \prod_{i=1}^{n_\phi} p_i 
    \\&= h_\ell \prod_{1\leq i \leq j \leq {n_\phi}} (\omega_i^2-\omega_j^2)  \prod_{k=1}^{n_\phi} \pd{^2 \phi_k(x_{\ell})}{x^2}.
\end{align*}
The first product is nonzero as long as $\omega_i \neq \omega_j$, which is always the case since each natural frequency is unique.
The determinant of $\mathcal{O}_C$ is thus zero if and only if the height of the beam, $h_\ell$, or the curvature of any mode shape is zero, that is if $p_i = \pd{^2 \phi_i(x_{\ell})}{x^2} = 0$ for any $i=1,2,\dots,{n_\phi}$.
\end{proof}

Note that Vandermonde matrices are generally ill-conditioned, so while the system would technically be observable with a single sensor, the estimation problem would likely be poorly conditioned.

\begin{figure}
    \centering
    \begin{subfigure}[b]{0.43\textwidth}
    \centering
       \includegraphics[width=\textwidth]{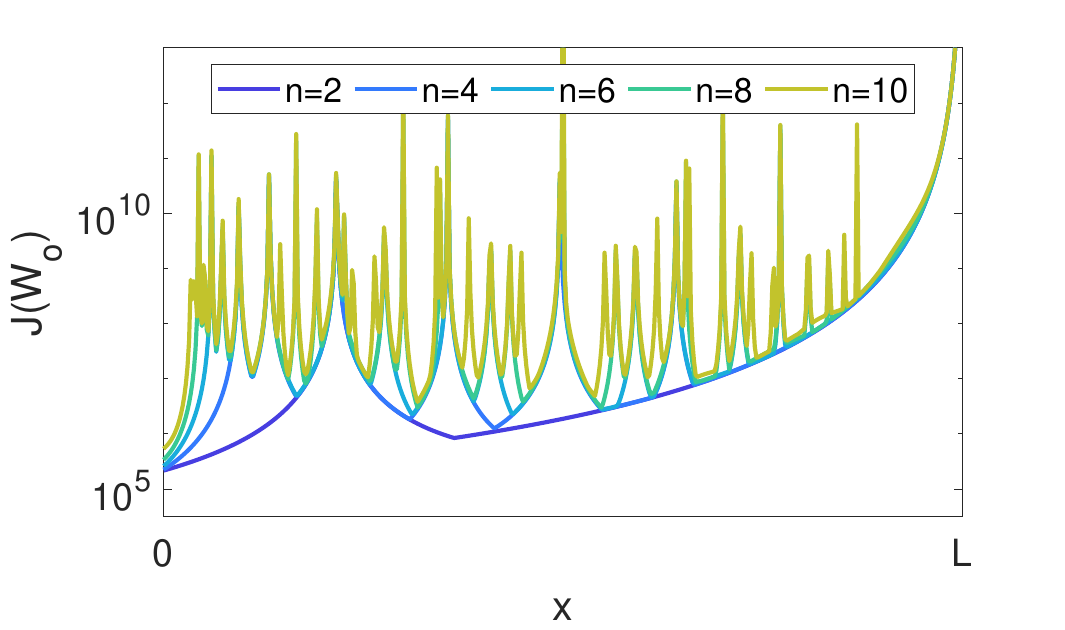}
    \caption{\centering Truncated System ($\Sigma_n$)}
    \label{fig:objFuncLinear}     
    \end{subfigure}
\hfill
    \begin{subfigure}[b]{0.43\textwidth}
       \includegraphics[width=\textwidth]{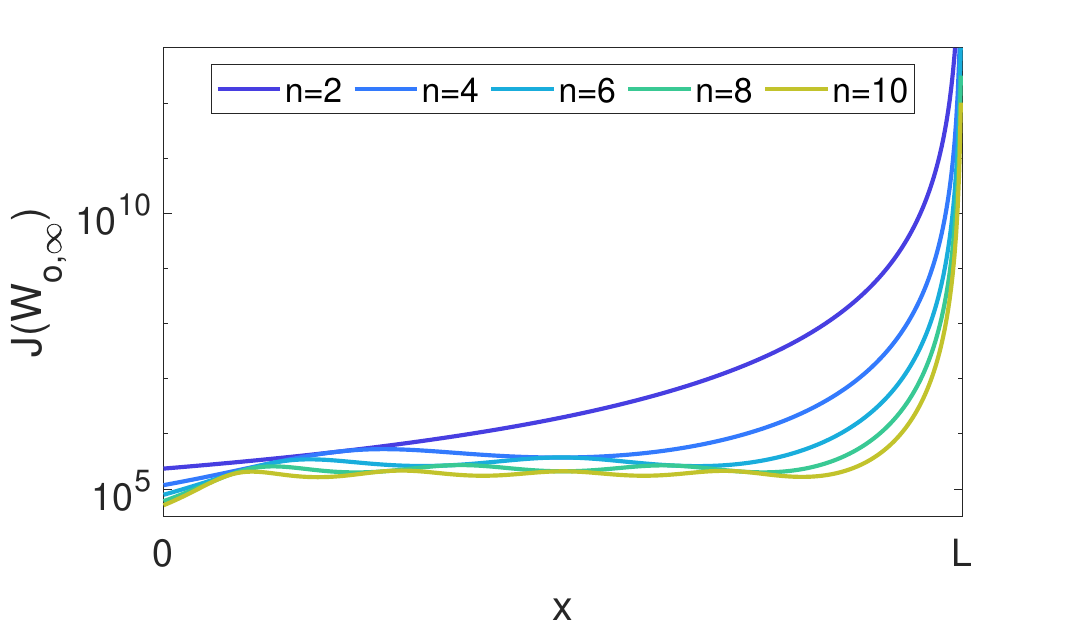}
    \caption{Continuum System ($\Sigma_\infty$)}
    \label{fig:objFuncContinuum}     
    \end{subfigure}
  \caption{The objective function of the observability Gramian, $J(W_o) = \kappa + w \nu$, plotted along the length of the beam for two to ten modes. \vspace{-1em}}
  \label{fig:objFuncBoth}
\end{figure}

\begin{figure}
    \centering
    \includegraphics[width = 0.49\textwidth]{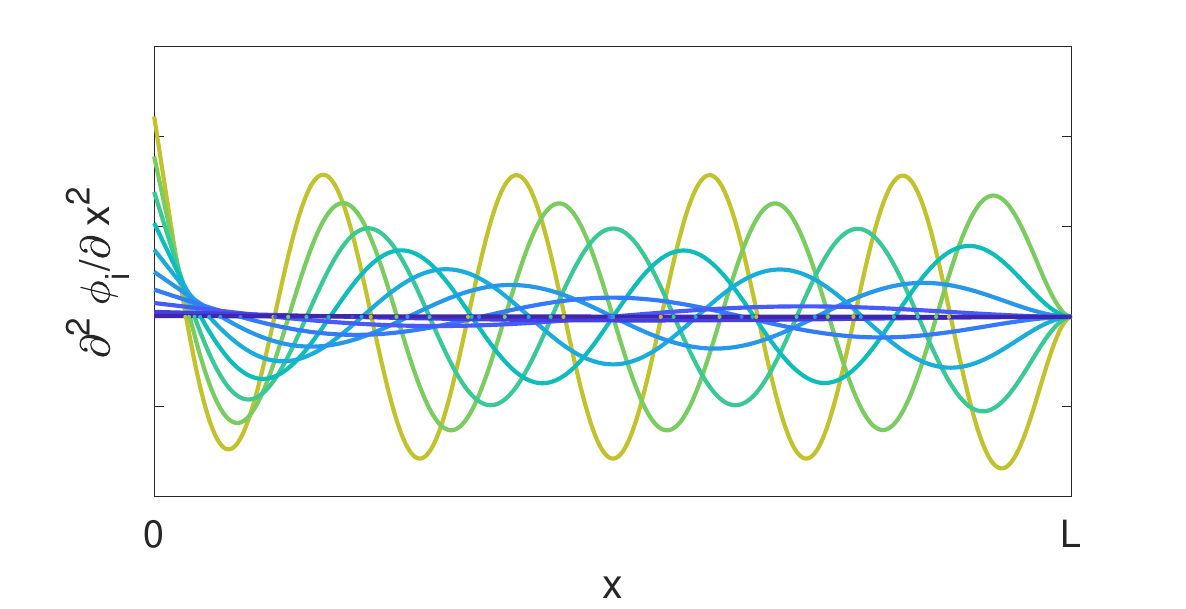}
    \caption{Second derivative of mode shapes with respect to $x$ for modes one through ten.}
    \label{fig:modes_strain}
\end{figure}

\begin{figure}
    \centering
    \includegraphics[height = 1.8in]{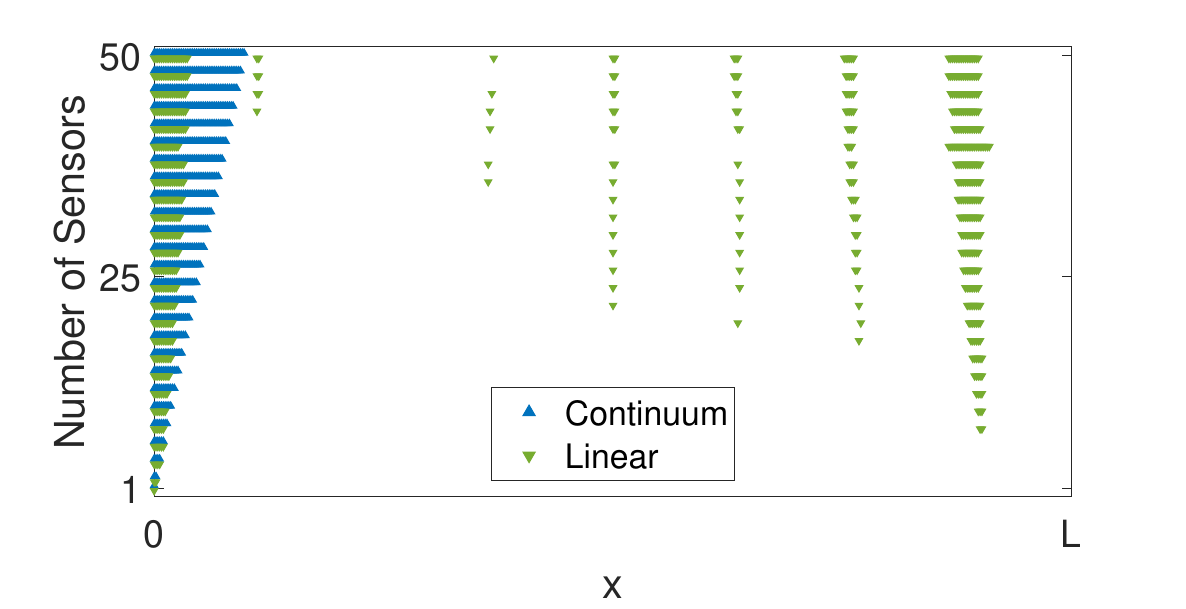}
    \caption{Optimal sensor locations for $n_\phi=8$ modes for the continuum and truncated systems ($\Sigma_\infty$ and $\Sigma_n$, respectively).}
    \label{fig:sensorPlacement}
\end{figure}

  \begin{figure}
    \centering
     \begin{subfigure}[b]{0.45\textwidth}
         \centering
         \includegraphics[height=1.8in]{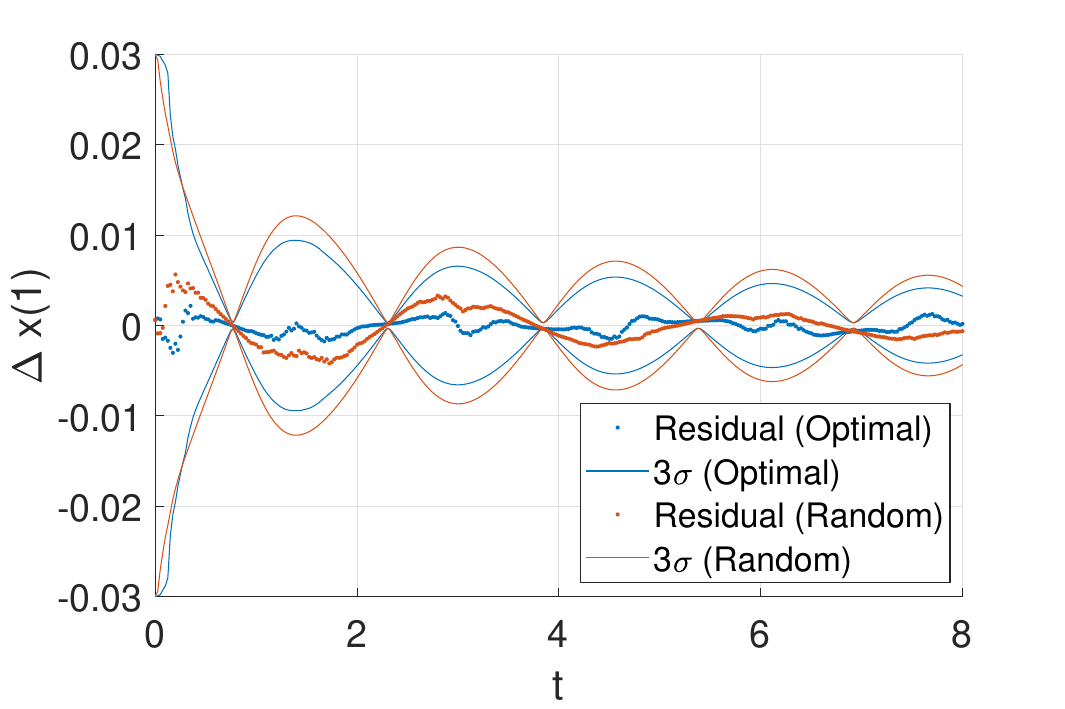}
         \caption{First mode, $\eta_1$}
         \label{fig:resid_mode1}
     \end{subfigure}
     \hfill
     \begin{subfigure}[b]{0.45\textwidth}
         \centering
         \includegraphics[height=1.8in]{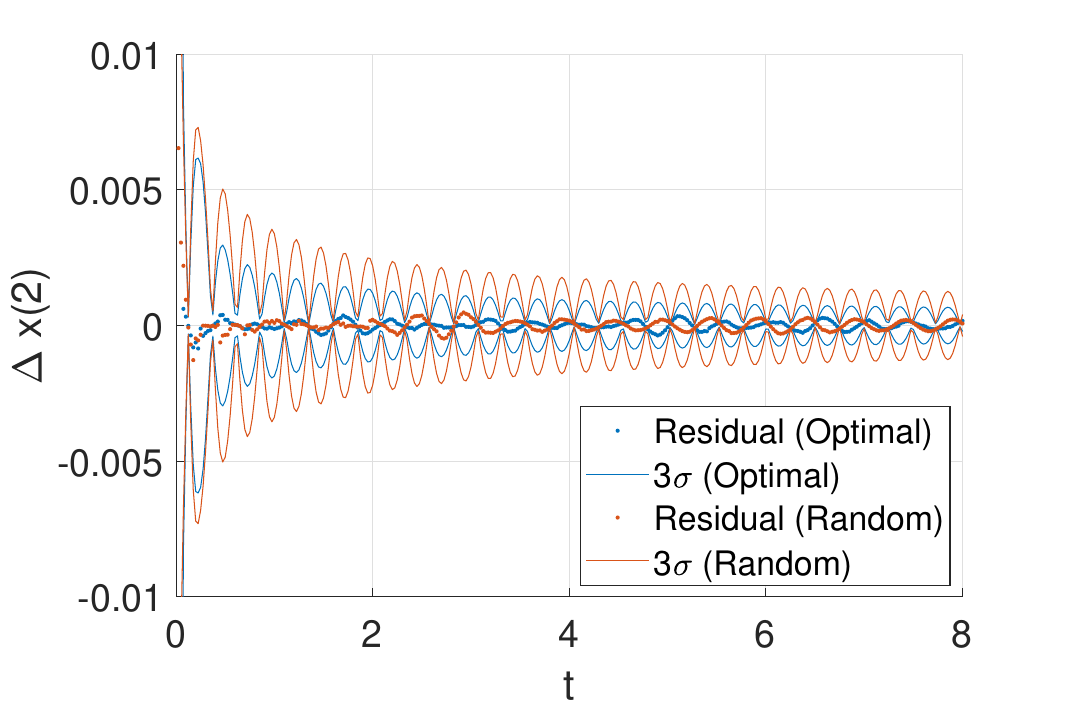}
         \caption{Second mode, $\eta_2$}
         \label{fig:resid_mode2}
     \end{subfigure}
        \caption{Residual error and $3\sigma$ covariance bounds from the UKF estimation of $\Sigma_n$ with $n=10$ modes based on measurements from ten strain sensors placed optimally (blue) and randomly (red).}
        \label{fig:optimal_results_resid}
\end{figure}

  \begin{figure}
         \centering
         \includegraphics[height=1.75in]{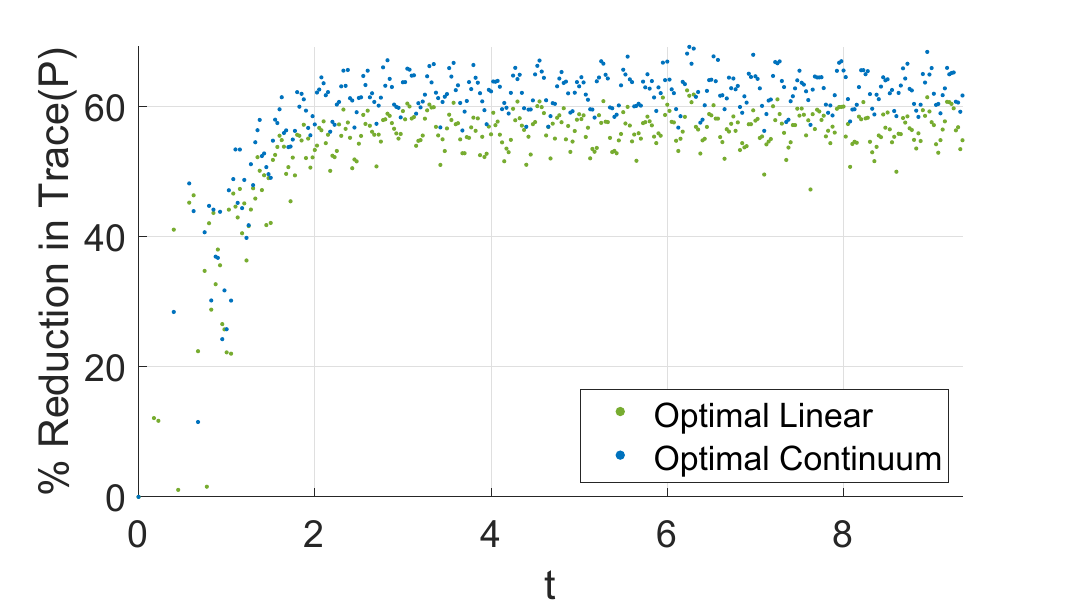}
        \caption{Comparison of the trace of the covariance matrix $P$ from the UKF estimation of $\Sigma_n$ with $n=10$ modes based on measurements from ten strain sensors placed optimally for $\Sigma_\infty$ (blue) and $\Sigma_n$ (green). \vspace{-1em}}
        \label{fig:optimal_results_cov}
\end{figure}   

\section{Numerical Results} \label{sec:numerical}

To evaluate the results of the previous sections, a rectangular cantilever beam was simulated in MATLAB.
The simulated beam was 5 mm thick, 20 mm wide, and 2m long with properties of aluminum (modulus of elasticity of 70 GPa and density of 2,700 kg/m$^3$) and discretized into $N=501$ sections along its length. The dynamics were approximated using the first $n_\phi=8$ modes. The empirical observability Gramians $W_o^\epsilon \in \mathbb{R}^{16 \times 16}$ and $W_{o,\infty}^\epsilon \in \mathbb{R}^{2 \times 2}$ were calculated for strain sensors at each of the 501 evenly distributed locations along the beam.

In order to compare the observability results of the continuum system to the finite dimensional truncated system, we must account for the differences in the dimensionality of these results.  We facilitate this comparison through the use of a scalar metric of the Gramian for each of the frameworks. Figure \ref{fig:objFuncBoth} shows the objective function, $J(W_o^\epsilon)$, from (\ref{eq:objFunc}) for the truncated and continuum cases; in both cases, the objective function is smallest near the fixed end of the beam where strain is the largest and becomes large at the free end where the strain goes to zero. The peaks in the plots of the truncated system (Fig. \ref{fig:objFuncLinear}) correspond to zeros of $\phi_{i,xx}(x)$ (Fig. \ref{fig:modes_strain}) where the system would not be observable with a single measurement, as proven in theorem \ref{thm:singleSensor}.

To determine the optimal sensor placement, the optimization problem \eqref{eq:optProb} was then solved using CVX \cite{cvx} for each of the systems for a maximum of one to fifty sensors; the optimal sensor placement is shown in Fig. \ref{fig:sensorPlacement}. The optimal locations for both the truncated and continuum systems include sensors near the fixed end of the beam $(x=0)$ where the strain energy is highest. For higher numbers of sensors, the truncated system also places sensors further out along the beam, possibly due to the energy associated with strain of some of the higher mode shapes or the structure of the Gramian (outer product of two terms vs n terms). For a nonlinear system, the results of the optimization may not be as intuitive as those of these relatively simple linear systems, but the same sensor placement problem can be posed for nonlinear systems: once the empirical observability Gramians have been constructed from the nonlinear system dynamics, the optimization algorithm follows the same process regardless of the underlying system dynamics. 

To compare the results of optimal to naïve sensor placement, an Unscented Kalman Filter (UKF) was used to estimate the states of the ODE system, $\eta_i(t)$ and $\dot{\eta}_i(t)$, for $i = 1,\dots,10$. Ten strain sensors were simulated with normally distributed measurement noise with a covariance $R = 10^{-4} I^{n\times n}$ and the estimate covariance was initialized as a diagonal matrix with $P_{0,\eta_i} = 10^{-2}$ and $P_{0,\dot{\eta_i}} = 10^{-4}$. Optimally placed sensors based on both $\Sigma_n$ and $\Sigma_\infty$ performed better than randomly placed sensors in terms of the error residuals and the error covariance. Plots of the residual errors and $3\sigma$ bounds for the first two states are shown in Fig. \ref{fig:optimal_results_resid}. The percent reduction of the trace of the covariance matrix $P$ of the optimally compared to randomly placed sensors is shown in \ref{fig:optimal_results_cov}; both sets of optimal sensors reduced the trace of the covariance by over 50\% as compared to the randomly place sensors.
                  
\section{Conclusions} \label{sec:future}

This work develops analytical and empirical observability Gramians for a PDE system that describes a freely vibrating cantilevered Euler-Bernoulli beam, uses those Gramians for optimal sensor placement, and compares the results with those of the finite approximation ODE.
While there are similarities between the Gramians for the continuum and finite approximation in terms of the chosen objective function, more work is needed to understand the underlying structural similarities and, perhaps more importantly, differences as we look to apply these techniques to more complicated (nonlinear and generally analytically intractable) systems such as aircraft or insect wings.
As an intermediate step, we plan to extend this work to a flat plate connected to a moving, rigid body and by including a damping terms.
With the addition of damping, the system can be described by a Riesz-spectral operator, and there are additional methods available for analyzing the infinite-dimensional system.

To generalize the application of the empirical observability Gramian, more work is required to determine what requirements there are for the initial perturbation functions to allow construction of the empirical Gramian.
In simulation, perturbing the system exactly by a mode shape is easily done; however, for systems that are to be analyzed by perturbing a physical model instead, such perturbations would likely be infeasible.

\section*{Acknowledgements} 

The authors would like to thank B. Boyac\i o\u{g}lu for his contributions to the development and editing of this manuscript.

\IEEEtriggeratref{5}
\bibliographystyle{ieeetr}
\bibliography{EulerBernoulliBeamRefs.bib}

\end{document}